\documentclass[11pt]{article}
\usepackage[T1]{fontenc}
\usepackage[margin=1in]{geometry}
\usepackage{graphicx}
\usepackage[longnamesfirst,sort]{natbib}
\usepackage{enumerate}
\usepackage{multirow}
\usepackage{booktabs}
\usepackage{colortbl}
\usepackage{amsmath}
\usepackage{amssymb}
\usepackage{amsthm}
\usepackage{xcolor}
\usepackage[hidelinks]{hyperref}
\hypersetup{
  pdftitle={A Bures-Wasserstein Formulation of Matrix Decomposition Structural Equation Modeling},
  pdfauthor={Naoto Yamashita},
  pdfkeywords={distance metric, consistency, asymptotic normality, covariance structure}
}
\def\diag{\mathrm{diag}}

\def\tr{\mathrm{tr}}

\def\argmin{\mathrm{argmin}}
\def\p{\prime}
\def\bfs{\boldsymbol}
\definecolor{colR1}{HTML}{000000}
\definecolor{colR2}{HTML}{000000} 
\definecolor{colR3}{HTML}{000000} 
\definecolor{colAE}{HTML}{000000} 
\newtheorem{thm}{Theorem}

\newtheorem{col}{Corollary}

\title{A Bures-Wasserstein Formulation of Matrix Decomposition Structural Equation Modeling}
\author{Naoto Yamashita\\
\normalsize Graduate School of Psychology, Kansai University, Japan\\
\normalsize \texttt{nyam@kansai-u.ac.jp}}
\date{\today}

\begin{document}
\maketitle

\begin{abstract}
Matrix decomposition SEM (MDSEM) is a data-matrix-based alternative to conventional covariance-based SEM, but its theoretical relationship to covariance-based SEM and the statistical properties of its estimator have remained unclear. 
We first reformulate MDSEM using a single loss function that integrates the measurement and structural models. 
We then prove that minimizing this loss over the model parameters is equivalent to minimizing the squared Bures–Wasserstein (BW) distance between the observed and model-implied covariance matrices, and use this equivalence to establish the estimator’s theoretical properties.
Specifically, the equivalence identifies the proposed estimator as both an MDSEM estimator and a covariance-based SEM estimator within the minimum discrepancy estimation framework, from which consistency, asymptotic normality, and standard errors are derived. 
Simulations show that its finite-sample performance is comparable to that of conventional estimators, including maximum likelihood, and that confidence-interval coverage is close to the nominal level. 
The estimator is also more numerically stable than conventional SEM estimators in small samples and under model misspecification. 
Thus, the BW formulation provides a theoretical foundation for MDSEM and a practically useful discrepancy function for covariance-based SEM.

\end{abstract}


\section{Introduction}
\label{sec:introduction}

Structural equation modeling (SEM) provides a flexible framework for modeling relationships among multiple observed and latent variables \citep{bentler1980multivariate,bentler1986structural,Bollen1989}.
The conventional formulation of SEM is covariance based: a model-implied covariance matrix ${\bf \Sigma}({\bfs\theta})$, parameterized by a vector of free parameters ${\bfs\theta}$, is fitted to the sample covariance matrix ${\bf S}_{\bf X} = n^{-1}{\bf X}^{\p}{\bf X}$, where ${\bf X}\ (n\times  p)$ is a column-centered data matrix of $n$ observations and $p$ variables.
The model parameters are estimated by minimizing a discrepancy function $d({\bf S}_{\bf X},{\bf \Sigma}({\bfs\theta}))$ with respect to ${\bfs\theta}$ \citep{kline2023}.

In contrast to covariance-based SEM, \citet{Yamashita2024} proposed matrix decomposition structural equation modeling (MDSEM), which fits an SEM model directly to the data matrix.
For ${\bf X}$, the original MDSEM criterion is
\begin{equation}
    L_{MD} =
    \lVert{\bf X} - ({\bf F}{\bf \Lambda}^{\p}
    + {\bf U}{\bf \Psi}^{1/2})\rVert^2
    + \alpha
    \lVert{\bf F} - {\bf FB}\rVert^2 .
    \label{eq:mdsem_loss}
\end{equation}
Here, ${\bf F}\ (n\times r)$ and ${\bf U}\ (n\times p)$ are score matrices of $r\ (<p)$ common and $p$ unique factors, respectively, ${\bf \Lambda}\ (p \times r)$ is a factor loading matrix, ${\bf B}\ (r \times r)$ is a path coefficient matrix, and ${\bf \Psi}^{1/2}\ (p\times p)$ is a diagonal matrix containing the square roots of the uniquenesses.
The factor score matrices are constrained as
\begin{equation}
    {\bf 1}_n^{\p}{\bf F} = {\bf 0}_r,\ 
    {\bf 1}_n^{\p}{\bf U} = {\bf 0}_p,\ 
    \diag(n^{-1}{\bf F}^{\p}{\bf F}) = {\bf I}_r,\
    n^{-1}{\bf U}^{\p}{\bf U} = {\bf I}_p,\
    {\bf F}^{\p}{\bf U} = {\bf O}_{r\times p}.
    \label{eq:mdsem_const}
\end{equation}
The first term in (\ref{eq:mdsem_loss}) measures the reconstruction error in the measurement equation, whereas the second measures the residual error in the structural relations among the latent variables.
The original MDSEM criterion combines these two losses through the tuning parameter $\alpha>0$.
Thus, MDSEM approximates the data matrix by an SEM model, whereas covariance-based SEM approximates the sample covariance matrix by a model-implied covariance matrix.

Despite their common objective of estimating SEM parameters, the theoretical relationship between MDSEM and covariance-based SEM has not been clearly established.
In particular, it remains unclear how an MDSEM solution can be interpreted within the conventional covariance-structure framework.
This difficulty stems from the apparent conceptual gap between their formulations: covariance-based SEM minimizes a discrepancy between covariance matrices, whereas MDSEM minimizes the matrix-decomposition criterion in (\ref{eq:mdsem_loss}).
Consequently, the validity of MDSEM has been supported primarily by parameter-recovery studies and empirical comparisons of its solutions with those obtained by covariance-based SEM, as reported in \citet{Yamashita2024}, rather than by an explicit covariance-discrepancy representation.
Moreover, asymptotic properties of the MDSEM estimator have not been established, although such properties have been studied extensively for covariance-based SEM; see \citet{joreskog1969, joreskog1970} for their classical examples.

Against this background, the central question of this study is whether the formulation of MDSEM admits an exact representation as a covariance discrepancy problem.
The original MDSEM criterion combines the residual losses from the measurement and structural equations through the tuning parameter $\alpha>0$ and therefore does not immediately yield a single covariance discrepancy corresponding to the entire SEM model.
We return to the measurement and structural equations underlying MDSEM and combine them into a single data-reconstruction equation.
This construction naturally yields an $\alpha$-free matrix-decomposition criterion.

We then show that the concentrated form of this $\alpha$-free criterion is exactly equal to the squared Bures-Wasserstein (BW) distance between the sample covariance matrix and the model-implied covariance matrix \citep{BhatiaJainLim2019}.
The BW distance thereby provides an explicit theoretical bridge between the data-matrix formulation of MDSEM and the covariance-structure formulation of conventional SEM.

This equivalence has several important consequences.
First, it places the MDSEM estimator within the minimum discrepancy function framework \citep{shapiro1983, shapiro1984, shapiro1985a, shapiro1985b}, from which its consistency and asymptotic normality can be established.
The resulting asymptotic covariance matrix also enables standard errors of the parameter estimates to be calculated.
Second, the model parameters can be estimated by directly minimizing the squared BW distance, providing an alternative to the alternating least-squares algorithm used for the original MDSEM.
Numerical studies and an empirical application evaluate the finite-sample performance and numerical stability of the proposed estimator.

The remainder of this article is organized as follows.
The next section presents the new formulation of MDSEM, establishes its connection to the squared BW distance, derives its asymptotic properties, and describes the parameter estimation and standard error calculations.
The subsequent sections report the numerical studies and empirical application.
The final section discusses the implications, limitations, and future directions of this study.

\section{Proposed Method and Its Asymptotic Properties}
\label{sec:proposed_method}

This section derives the covariance discrepancy corresponding to the matrix-decomposition formulation of MDSEM.
We first construct an integrated matrix-decomposition criterion from the measurement and structural equations and then show that its concentrated form is the squared BW distance.

\subsection{Derivation of an Integrated Matrix-Decomposition Criterion}
\label{subsec:formulation_of_mdsem}

The two terms in (\ref{eq:mdsem_loss}) correspond to residual losses from the measurement and structural equations.
Introducing residual matrices ${\bf E}_x\ (n\times p)$ and ${\bf E}_f\ (n\times r)$, these equations can be written as
\begin{equation}
    {\bf X} = {\bf F}{\bf \Lambda}^{\p}
    + {\bf U}{\bf \Psi}^{1/2} + {\bf E}_x
    \label{eq:mdsem_measurement_eq}
\end{equation}
and
\begin{equation}
    {\bf F} = {\bf FB} + {\bf E}_f .
    \label{eq:mdsem_structural_eq}
\end{equation}
The zero and free patterns in ${\bf \Lambda}$ and ${\bf B}$ encode the measurement and structural relations.
Thus, the original MDSEM objective can be viewed as minimizing
\begin{equation}
    \lVert{\bf E}_x\rVert^2+\alpha\lVert{\bf E}_f\rVert^2,
\end{equation}
that is, a weighted sum of the squared residuals in the two equations.

Because the diagonal elements of ${\bf B}$ are usually fixed to zero to avoid self-regression, and because ${\bf I}_r-{\bf B}$ is assumed to be nonsingular, (\ref{eq:mdsem_structural_eq}) can be rewritten as
\begin{equation}
    {\bf F} = {\bf E}_f({\bf I}_r-{\bf B})^{-1}.
    \label{eq:mdsem_structural_form2}
\end{equation}
Substituting this expression into (\ref{eq:mdsem_measurement_eq}) gives
\begin{equation}
    {\bf X}
    =
    {\bf E}_f({\bf I}_r-{\bf B})^{-1}{\bf \Lambda}^{\p}
    +{\bf U}{\bf \Psi}^{1/2}
    +{\bf E}_x .
    \label{eq:mdsem_measurement_form2}
\end{equation}
Thus, instead of minimizing a weighted sum of $\lVert{\bf E}_x\rVert^2$ and $\lVert{\bf E}_f\rVert^2$, we can define an $\alpha$-free reconstruction loss from (\ref{eq:mdsem_measurement_form2}).
To express this criterion in a standard matrix-decomposition form, write ${\bf E}_f=\tilde{\bf E}_f{\bf T}$, where $n^{-1}\tilde{\bf E}_f^{\p}\tilde{\bf E}_f={\bf I}_r$ and ${\bf T}$ is a nonsingular matrix satisfying $\diag({\bf T}^{\p}{\bf T})={\bf I}_r$.
The resulting loss is
\begin{eqnarray}
    L_n({\bf Z},{\bf \Xi})
    &=&
    \left\lVert
    {\bf X}
    -
    \left\{
    \tilde{\bf E}_f{\bf T}({\bf I}_r-{\bf B})^{-1}{\bf \Lambda}^{\p}
    +
    {\bf U}{\bf \Psi}^{1/2}
    \right\}
    \right\rVert^2
    =
    \left\lVert
    {\bf X}
    -
    {\bf Z}{\bf \Xi}^{\p}
    \right\rVert^2 ,
    \label{eq:mdsem_alpha_free_loss}
\end{eqnarray}
where ${\bf Z}=[\tilde{\bf E}_f,{\bf U}]$.
As in (\ref{eq:mdsem_const}), the component score matrices collected in ${\bf Z}$ are centered, standardized, and mutually orthogonal, so that ${\bf 1}_n^{\p}{\bf Z}={\bf 0}_{r+p}^{\p}$ and $n^{-1}{\bf Z}^{\p}{\bf Z}={\bf I}_{r+p}$.
The block matrix ${\bf \Xi}$ is
\begin{equation}
    {\bf \Xi}
    =
    \left[
    {\bf \Lambda}({\bf I}_r-{\bf B})^{-1\p}{\bf T}^{\p}
    \quad
    {\bf \Psi}^{1/2}
    \right].
    \label{eq:mdsem_block_matrix}
\end{equation}
Because ${\bf E}_f=\tilde{\bf E}_f{\bf T}$ and $n^{-1}\tilde{\bf E}_f^{\p}\tilde{\bf E}_f={\bf I}_r$, the covariance matrix of the structural disturbances is
$
    n^{-1}{\bf E}_f^{\p}{\bf E}_f
    =
    {\bf T}^{\p}{\bf T}.
$
We denote this matrix by ${\bf \Omega}$, so that ${\bf \Omega}={\bf T}^{\p}{\bf T}$.
This loss without $\alpha$ provides the matrix-decomposition representation whose concentrated form is examined in the next subsection.

The same representation also gives the factor covariance matrix.
From (\ref{eq:mdsem_structural_form2}),
\begin{eqnarray}
    {\bf \Phi}
    &=&
    n^{-1}{\bf F}^{\p}{\bf F} \nonumber\\
    &=&
    ({\bf I}_r-{\bf B})^{-1\p}
    {\bf T}^{\p}
    (n^{-1}\tilde{\bf E}_f^{\p}\tilde{\bf E}_f)
    {\bf T}
    ({\bf I}_r-{\bf B})^{-1} \nonumber\\
    &=&
    ({\bf I}_r-{\bf B})^{-1\p}
    {\bf \Omega}
    ({\bf I}_r-{\bf B})^{-1}.
    \label{eq:factor_correlation}
\end{eqnarray}
The loss in (\ref{eq:mdsem_alpha_free_loss}) is a matrix decomposition factor analysis (MDFA; \citeauthor{AdachiTrendafilov2018a}, \citeyear{AdachiTrendafilov2018a})-type reconstruction loss in which the loading block is structurally restricted by ${\bf \Lambda}$, ${\bf B}$, and ${\bf T}$.

Although not numerically equivalent to (\ref{eq:mdsem_loss}), the proposed criterion in (\ref{eq:mdsem_alpha_free_loss}) retains the defining idea of MDSEM by integrating the measurement and structural equations underlying the two terms of (\ref{eq:mdsem_loss}) into a single data-matrix decomposition.

\subsection{Connection to the Squared Bures-Wasserstein Distance}
\label{subsec:bw_connection}

The preceding formulation expresses MDSEM through an MDFA-type reconstruction loss.
\citet{terada2025} showed that, for MDFA, the concentrated matrix-decomposition loss is equivalent to the squared BW distance between the sample covariance matrix and the model-implied covariance matrix.
Specifically, after minimizing the MDFA reconstruction loss over the factor score matrix, the resulting concentrated loss can be written as a squared BW discrepancy.
Applying this result to (\ref{eq:mdsem_alpha_free_loss}) yields the following corollary.

\begin{col}
Let ${\bfs \theta}$ collect the free parameters in
${\bf \Lambda}$, ${\bf B}$, ${\bf \Omega}$, and ${\bf \Psi}$ in a suitable order.
For any matrix ${\bf T}$ such that ${\bf T}^{\p}{\bf T}={\bf \Omega}$, let ${\bf\Xi}({\bfs\theta})$ be defined as in (\ref{eq:mdsem_block_matrix}).
Then the concentrated loss
\begin{equation}
    L_n({\bfs \theta})
    =
    \min_{{\bf Z}}
    L_n({\bf Z},{\bf \Xi}({\bfs \theta}))
\end{equation}
satisfies
\begin{equation}
    L_n({\bfs \theta})
    =
    d^2_{BW}
    \left(
    {\bf S}_{\bf X},
    {\bf \Sigma}({\bfs \theta})
    \right),
    \label{eq:mdsem_bw_equivalence}
\end{equation}
where
\begin{eqnarray}
    {\bf \Sigma}({\bfs \theta})
    =
    {\bf \Xi}({\bfs \theta}){\bf \Xi}({\bfs \theta})^{\p} 
    =
    {\bf \Lambda}({\bf I}_r-{\bf B})^{-1\p}
    {\bf \Omega}
    ({\bf I}_r-{\bf B})^{-1}
    {\bf \Lambda}^{\p}
    +
    {\bf \Psi}
    \label{eq:mdsem_implied_covariance}
\end{eqnarray}
and $d_{BW}^2({\bf P}, {\bf Q})$ is the squared BW distance between two positive definite matrices ${\bf P}$ and ${\bf Q}$ defined as
\begin{equation}
    d_{BW}^2({\bf P}, {\bf Q})
    =
    \tr({\bf P}) + \tr({\bf Q})
    -2
    \tr
    \left\{
    (
    {\bf P}^{1/2}{\bf Q}{\bf P}^{1/2}
    )^{1/2}
    \right\}.
\end{equation}
\end{col}

\begin{proof}
Because $n^{-1}{\bf Z}^{\p}{\bf Z}={\bf I}_{r+p}$, Proposition 3.1 of \citet{terada2025} applies to the block matrix ${\bf \Xi}({\bfs \theta})$ in (\ref{eq:mdsem_alpha_free_loss}), yielding the stated result.\qed
\end{proof}

In SEM, substantive interest lies in the model parameters rather than in the auxiliary score matrix ${\bf Z}$.
Because joint and concentrated minimization yield the same parameter estimates, ${\bf Z}$ can be concentrated out without altering the estimation problem.

Moreover, because $n^{-1}{\bf Z}^{\p}{\bf Z}={\bf I}_{r+p}$, the covariance matrix of the reconstructed data matrix ${\bf Z}{\bf \Xi}^{\p}$ is
$
    n^{-1}
    ({\bf Z}{\bf \Xi}^{\p})^{\p}
    ({\bf Z}{\bf \Xi}^{\p})
    =
    {\bf \Xi}{\bf \Xi}^{\p}
    =
    {\bf \Sigma}({\bfs\theta}).
$
Thus, the corollary shows that the concentrated MDSEM loss can be interpreted as the squared BW distance between the sample covariance matrix and the model-implied covariance structure.

Because this covariance structure, and hence the concentrated loss, depends on ${\bf T}$ only through ${\bf T}^{\p}{\bf T}={\bf \Omega}$, we treat ${\bf \Omega}$, rather than ${\bf T}$ itself, as the structural-disturbance parameter.
Accordingly, the substantive model parameters are ${\bf \Lambda}$, ${\bf B}$, ${\bf \Omega}$, and ${\bf \Psi}$, whereas ${\bf T}$ is an auxiliary matrix used only in the matrix-decomposition representation.

Therefore, the proposed MDSEM estimator can be written as
\begin{equation}
    \hat{\bfs\theta}_n
    =
    \argmin_{{\bfs\theta}\in\mathcal{H}}
    d^2_{BW}
    \left(
    {\bf S}_{\bf X},
    {\bf \Sigma}({\bfs\theta})
    \right),
    \label{eq:mdsem_bw_estimator}
\end{equation}
where $\mathcal{H}$ denotes the parameter space.
Thus, the proposed MDSEM estimator is a minimum discrepancy function estimator with the squared BW distance as its discrepancy function.
This representation provides the basis for the asymptotic theory and standard error calculation developed below.

\subsection{Asymptotic Properties of the MDSEM Estimator}
\label{subsec:asymptotic_properties}

We next establish the asymptotic properties of the proposed MDSEM estimator by treating it as a minimum discrepancy function estimator.
Let ${\bfs\theta}_0\in\mathcal{H}$ denote a true parameter vector and let
\begin{equation}
    {\bf \Sigma}_0
    =
    {\bf \Sigma}({\bfs\theta}_0)
\end{equation}
be the corresponding population covariance matrix.
The population counterpart of the sample criterion in (\ref{eq:mdsem_bw_estimator}) is
\begin{equation}
    L({\bfs\theta})
    =
    d_{BW}^2
    \left(
    {\bf \Sigma}_0,
    {\bf \Sigma}({\bfs\theta})
    \right).
    \label{eq:mdsem_population_criterion}
\end{equation}
Because the BW distance is a metric \citep[see, e.g.,][]{BhatiaJainLim2019},
\begin{equation}
    L({\bfs\theta})=0
    \quad\Longleftrightarrow\quad
    {\bf\Sigma}({\bfs\theta})={\bf\Sigma}_0.
\end{equation}
Accordingly, define the set of population minimizers by
\begin{equation}
    \mathcal{H}_0
    =
    \left\{
    {\bfs\theta}\in\mathcal{H}:
    {\bf \Sigma}({\bfs\theta})={\bf \Sigma}_0
    \right\}.
    \label{eq:mdsem_true_parameter_set}
\end{equation}
If the SEM model is globally identifiable at ${\bfs\theta}_0$, then $\mathcal{H}_0=\{{\bfs\theta}_0\}$, a singleton.
Thus, the distinction between convergence to the set $\mathcal{H}_0$ and convergence to the unique parameter vector ${\bfs\theta}_0$ depends on whether global identifiability is imposed.

Once the proposed MDSEM estimator is represented as a minimum discrepancy function estimator, its consistency follows from the standard asymptotic theory of minimum discrepancy estimation \citep{shapiro1983}; see also Theorem 3.4 of \citet{terada2025} for the corresponding result for MDFA.
The following theorem records this implication for the proposed MDSEM estimator.

\begin{thm}[Consistency of the proposed MDSEM estimator]
Assume that
\begin{enumerate}
    \item ${\bf x}_1,\ldots,{\bf x}_n$ are independent and identically
    distributed with covariance matrix ${\bf \Sigma}_0$ and finite
    second moments;
    \item the parameter space $\mathcal{H}$ is compact; and
    \item ${\bf \Sigma}({\bfs\theta})$ is continuous in
    ${\bfs\theta}\in\mathcal{H}$.
\end{enumerate}
Then the MDSEM estimator defined in
(\ref{eq:mdsem_bw_estimator}) satisfies
\begin{equation}
    L(\hat{\bfs\theta}_n)
    \xrightarrow{\mathrm{a.s.}} 0
\end{equation}
and
\begin{equation}
    \inf_{{\bfs\theta}\in\mathcal{H}_0}
    \left\lVert
    \hat{\bfs\theta}_n-{\bfs\theta}
    \right\rVert
    \xrightarrow{\mathrm{a.s.}} 0.
    \label{eq:mdsem_set_consistency}
\end{equation}
\end{thm}

\begin{proof}
By the strong law of large numbers,
\begin{equation}
    {\bf S}_{\bf X}
    \xrightarrow{\mathrm{a.s.}}
    {\bf \Sigma}_0.
\end{equation}
The continuity of the squared BW distance and of
${\bf \Sigma}({\bfs\theta})$, together with the compactness of
$\mathcal{H}$, implies
\begin{equation}
    \sup_{{\bfs\theta}\in\mathcal{H}}
    \left|
    L_n({\bfs\theta})-L({\bfs\theta})
    \right|
    \xrightarrow{\mathrm{a.s.}} 0.
\end{equation}
Because $L({\bfs\theta})\geq 0$ and its set of minimizers is
$\mathcal{H}_0$, the standard consistency argument for minimum
discrepancy function estimators gives
(\ref{eq:mdsem_set_consistency}).
Moreover, the uniform convergence and
$L_n(\hat{\bfs\theta}_n)\leq L_n({\bfs\theta}_0)$ imply
$L(\hat{\bfs\theta}_n)\to 0$ almost surely.\qed
\end{proof}

\begin{col}
Under the conditions of the preceding theorem, if the SEM model is
globally identifiable at ${\bfs\theta}_0$, then
\begin{equation}
    \hat{\bfs\theta}_n
    \xrightarrow{\mathrm{a.s.}}
    {\bfs\theta}_0.
    \label{eq:mdsem_parameter_consistency}
\end{equation}
\end{col}

\begin{proof}
Global identifiability implies
$\mathcal{H}_0=\{{\bfs\theta}_0\}$.
The result therefore follows directly from
(\ref{eq:mdsem_set_consistency}).\qed
\end{proof}

To establish asymptotic normality, stronger local regularity conditions are required.
The principal matrix square-root map is smooth on the cone of positive definite matrices \citep{Freidlin1968}.
Consequently, if ${\bf \Sigma}({\bfs\theta})$ is twice continuously differentiable and remains positive definite in a neighborhood of ${\bfs\theta}_0$, the population criterion $L({\bfs\theta})$ is also twice continuously differentiable in that neighborhood; see Proposition 3.3 of \citet{terada2025}.

Like the preceding consistency result, the following result is an application of the standard asymptotic theory for minimum discrepancy function estimators \citep{shapiro1983} to the proposed MDSEM estimator.

\begin{thm}[Asymptotic normality of the proposed MDSEM estimator]
Assume the conditions of the preceding theorem and corollary.
Furthermore, assume that
\begin{enumerate}
    \item
    $\mathrm{E}(\lVert{\bf x}_i\rVert^4)<\infty$;
    \item
    ${\bfs\theta}_0$ is an interior point of $\mathcal{H}$;
    \item
    ${\bf \Sigma}({\bfs\theta})$ is twice continuously differentiable
    and positive definite in a neighborhood of ${\bfs\theta}_0$; and
    \item
    the Hessian matrix
    \begin{equation}
        {\bf H}_0
        =
        \left.
        \frac{\partial^2 L({\bfs\theta})}
        {\partial{\bfs\theta}\partial{\bfs\theta}^{\p}}
        \right|_{{\bfs\theta}={\bfs\theta}_0}
    \end{equation}
    is nonsingular.
\end{enumerate}
Then
\begin{equation}
    \sqrt{n}
    \left(
    \hat{\bfs\theta}_n-{\bfs\theta}_0
    \right)
    \xrightarrow{d}
    N({\bf 0},{\bf V}),
    \label{eq:mdsem_asymptotic_normality}
\end{equation}
where ${\bf V}$ is the asymptotic covariance matrix of the proposed
MDSEM estimator.
\end{thm}

\begin{proof}
By the preceding corollary,
$\hat{\bfs\theta}_n\xrightarrow{\mathrm{a.s.}}{\bfs\theta}_0$.
Moreover, the finite fourth-moment condition implies
\begin{equation}
    \sqrt{n}
    \left\{
    \operatorname{vech}({\bf S}_{\bf X})
    -
    \operatorname{vech}({\bf \Sigma}_0)
    \right\}
    \xrightarrow{d}
    N({\bf 0},{\bf \Gamma}),
    \label{eq:sample_covariance_asymptotic_normality}
\end{equation}
where ${\bf \Gamma}$ is the asymptotic covariance matrix of
$\operatorname{vech}({\bf S}_{\bf X})$.

The smoothness of the squared BW discrepancy and of
${\bf \Sigma}({\bfs\theta})$ ensures that the discrepancy function is twice continuously differentiable in a neighborhood of $({\bf \Sigma}_0,{\bfs\theta}_0)$.
Together with the consistency of $\hat{\bfs\theta}_n$, the asymptotic
normality in
(\ref{eq:sample_covariance_asymptotic_normality}), and the
nonsingularity of ${\bf H}_0$, the standard asymptotic normality
theorem for minimum discrepancy function estimators \citep{shapiro1983} yields
(\ref{eq:mdsem_asymptotic_normality}).\qed
\end{proof}

Thus, representing the proposed MDSEM estimator as a minimum discrepancy function estimator provides both consistency and asymptotic normality under standard regularity conditions.
These results do not require a new general asymptotic theory; rather, they establish that the existing theory of minimum discrepancy estimation applies to the proposed MDSEM estimator through its BW representation.
The asymptotic covariance matrix ${\bf V}$ also provides the basis for calculating standard errors without relying on computationally intensive bootstrap procedures.
The estimation procedure and the calculation of ${\bf V}$ are described in the following subsections.

\subsection{Parameter Estimation}
\label{subsec:parameter_estimation}
As shown in (\ref{eq:mdsem_bw_estimator}), the proposed MDSEM estimator can be obtained by directly minimizing
\[
    d_{BW}^2
    \left(
    {\bf S}_{\bf X},
    {\bf \Sigma}({\bfs\theta})
    \right)
\]
with respect to the free parameter vector ${\bfs\theta}$ over $\mathcal{H}$.
This formulation avoids alternating optimization over the score matrix ${\bf Z}$ and the model parameters, as employed in the original MDSEM.
It also allows the estimator to be computed using a general-purpose numerical optimization algorithm.
In this study, we use the L-BFGS-B \citep{ByrdLuNocedalZhu1995} algorithm.

The nonnegativity of the uniqueness variances can be imposed through a square-root parameterization.
Let $\psi_j$ denote the uniqueness variance of the $j$th observed variable and introduce an unconstrained computational parameter $\kappa_j$ such that
\begin{equation}
    \psi_j=\kappa_j^2,
    \qquad j=1,\ldots,p.
\end{equation}
Accordingly,
\begin{equation}
    {\bf \Psi}
    =
    \diag
    \left(
    \kappa_1^2,\ldots,\kappa_p^2
    \right).
\end{equation}
Numerical optimization is performed with respect to $\kappa_1,\ldots,\kappa_p$, whereas the uniqueness variances $\psi_1,\ldots,\psi_p$ are treated as the model parameters for reporting and statistical inference.
This parameterization ensures nonnegative uniqueness estimates.
It can also be used in conventional covariance-based SEM and is not specific to MDSEM.

The gradient of the squared BW distance with respect to its second matrix argument is
\begin{equation}
    {\bf G}
    \left(
    {\bf S}_{\bf X},
    {\bf \Sigma}({\bfs\theta})
    \right)
    =
    {\bf I}_p
    -
    {\bf S}_{\bf X}^{1/2}
    \left\{
    {\bf S}_{\bf X}^{1/2}
    {\bf \Sigma}({\bfs\theta})
    {\bf S}_{\bf X}^{1/2}
    \right\}^{-1/2}
    {\bf S}_{\bf X}^{1/2},
    \label{eq:bw_gradient_sigma}
\end{equation}
provided that the matrices involved are positive definite \citep{BhatiaJainLim2019}.
Therefore, the derivative with respect to the $k$th free parameter is obtained by the chain rule as
\begin{equation}
    \frac{\partial L_n({\bfs\theta})}{\partial\theta_k}
    =
    \tr
    \left\{
    {\bf G}
    \left(
    {\bf S}_{\bf X},
    {\bf \Sigma}({\bfs\theta})
    \right)
    \frac{\partial{\bf \Sigma}({\bfs\theta})}
    {\partial\theta_k}
    \right\}.
    \label{eq:bw_gradient_theta}
\end{equation}
Because numerical optimization uses the computational parameter $\kappa_j$, defined by $\psi_j=\kappa_j^2$, instead of the uniqueness variance $\psi_j$, the gradient is transformed by the chain rule.
In particular,
\begin{equation}
    \frac{\partial L_n}{\partial\kappa_j}
    =
    2\kappa_j
    \frac{\partial L_n}{\partial\psi_j},
    \qquad j=1,\ldots,p.
\end{equation}
The resulting gradient in the computational parameterization is supplied to the L-BFGS-B algorithm.

The latent variables admit the usual change of scale.
For any positive diagonal matrix ${\bf D}$, define
\begin{eqnarray}
    {\bf \Lambda}_{s}
    =
    {\bf \Lambda}{\bf D}, \ 
    {\bf B}_{s}
    =
    {\bf D}{\bf B}{\bf D}^{-1}, \ 
    {\bf \Omega}_{s}
    =
    {\bf D}^{-1}{\bf \Omega}{\bf D}^{-1}.
    \label{eq:mdsem_scaling_transformation}
\end{eqnarray}
The corresponding factor covariance matrix is
\begin{equation}
    {\bf \Phi}_{s}
    =
    {\bf D}^{-1}{\bf \Phi}{\bf D}^{-1}.
\end{equation}
These transformations leave the model-implied covariance matrix unchanged because
\begin{equation}
    {\bf \Lambda}_{s}
    {\bf \Phi}_{s}
    {\bf \Lambda}_{s}^{\p}
    +
    {\bf \Psi}
    =
    {\bf \Lambda}{\bf \Phi}{\bf \Lambda}^{\p}
    +
    {\bf \Psi}.
\end{equation}

During estimation, the scale is determined by the normalization $\diag({\bf \Omega})={\bf I}_r$.
For reporting standardized solutions, we instead set
\begin{equation}
    {\bf D}
    =
    \diag({\bf \Phi})^{1/2},
\end{equation}
so that $\diag({\bf \Phi}_{s})={\bf I}_r$.
The transformations in (\ref{eq:mdsem_scaling_transformation}) are then applied to ${\bf \Lambda}$, ${\bf B}$, and ${\bf \Omega}$.
Thus, the estimation normalization and the subsequent standardization for reporting purposes are distinguished explicitly.

\subsection{Standard Errors}
\label{subsec:standard_errors}

The asymptotic normality established above provides a basis for calculating standard errors of the proposed MDSEM estimator.
For statistical inference, let ${\bfs\theta}$ denote the parameter vector containing the uniqueness variances $\psi_1,\ldots,\psi_p$, rather than the computational parameters $\kappa_1,\ldots,\kappa_p$ used in numerical optimization.

An estimator of the asymptotic covariance matrix ${\bf V}$ in (\ref{eq:mdsem_asymptotic_normality}) is
\begin{eqnarray}
    \hat{\bf V}
    &=&
    (\hat{\bf\Delta}^{\p}\hat{\bf W}\hat{\bf\Delta})^{-1}
    \hat{\bf\Delta}^{\p}\hat{\bf W}
    \hat{\bf\Gamma}
    \hat{\bf W}\hat{\bf\Delta}
    (\hat{\bf\Delta}^{\p}\hat{\bf W}\hat{\bf\Delta})^{-1}.
    \label{eq:mdsem_asymptotic_covariance}
\end{eqnarray}
Accordingly, the estimated standard error of $\hat{\theta}_{k}$ is
\begin{equation}
    \widehat{\operatorname{se}}(\hat{\theta}_{k})
    =
    \left(
    \frac{\hat V_{kk}}{n}
    \right)^{1/2}.
    \label{eq:mdsem_standard_error}
\end{equation}

To define the matrices in (\ref{eq:mdsem_asymptotic_covariance}), let
\begin{equation}
    {\bfs\xi}({\bfs\theta})
    =
    \operatorname{vech}
    \left\{
    {\bf\Sigma}({\bfs\theta})
    \right\}.
\end{equation}
The matrix $\hat{\bf\Delta}$ is the Jacobian of the model-implied covariance vector evaluated at $\hat{\bfs\theta}_n$:
\begin{equation}
    \hat{\bf\Delta}
    =
    \left.
    \frac{\partial{\bfs\xi}({\bfs\theta})}
    {\partial{\bfs\theta}^{\p}}
    \right|_{{\bfs\theta}=\hat{\bfs\theta}_n}.
    \label{eq:mdsem_delta_matrix}
\end{equation}

To define $\hat{\bf W}$, regard ${\bfs\xi}=\operatorname{vech}({\bf\Sigma})$ as an unrestricted covariance coordinate.
Then
\begin{equation}
    \hat{\bf W}
    =
    \left.
    \frac{1}{2}
    \frac{\partial^2
    d_{BW}^2({\bf S}_{\bf X},{\bf\Sigma})}
    {\partial{\bfs\xi}\partial{\bfs\xi}^{\p}}
    \right|_{{\bf\Sigma}={\bf S}_{\bf X}}.
    \label{eq:mdsem_weight_matrix}
\end{equation}

Finally, $\hat{\bf\Gamma}$ estimates the asymptotic covariance matrix of $\operatorname{vech}({\bf S}_{\bf X})$.
For centered observations ${\bf x}_1,\ldots,{\bf x}_n$, it is estimated by
\begin{eqnarray}
    \hat{\bf\Gamma}
    &=&
    \frac{1}{n}
    \sum_{i=1}^{n}
    \left\{
    \operatorname{vech}({\bf x}_i{\bf x}_i^{\p})
    -
    \operatorname{vech}({\bf S}_{\bf X})
    \right\}
    \left\{
    \operatorname{vech}({\bf x}_i{\bf x}_i^{\p})
    -
    \operatorname{vech}({\bf S}_{\bf X})
    \right\}^{\p}.
    \label{eq:mdsem_gamma_estimator}
\end{eqnarray}
The matrices $\hat{\bf\Delta}$ and $\hat{\bf W}$ can be evaluated numerically.
In the numerical studies and empirical applications below, numerical differentiation is used to calculate these matrices and the resulting standard errors and confidence intervals.

Although numerical optimization is performed using the computational parameters $\kappa_1,\ldots,\kappa_p$, defined by $\psi_j=\kappa_j^2$, statistical inference concerns the uniqueness variances $\psi_1,\ldots,\psi_p$.
Let ${\bfs\theta}_{\kappa}$ denote the computational version of ${\bfs\theta}$ in which the uniqueness variances $\psi_1,\ldots,\psi_p$ are replaced by $\kappa_1,\ldots,\kappa_p$.
The two parameterizations are related by
\begin{equation}
    {\bfs\theta}
    =
    {\bf g}({\bfs\theta}_{\kappa}),
\end{equation}
where ${\bf g}$ leaves all other parameters unchanged and maps $\kappa_j$ to $\psi_j=\kappa_j^2$.

After numerical optimization, the estimator is expressed in the inferential parameterization as
\begin{equation}
    \hat{\bfs\theta}_n
    =
    {\bf g}(\hat{\bfs\theta}_{\kappa,n}).
\end{equation}
The matrices $\hat{\bf\Delta}$ and $\hat{\bf V}$ in (\ref{eq:mdsem_delta_matrix}) and (\ref{eq:mdsem_asymptotic_covariance}) are then evaluated directly with respect to ${\bfs\theta}$.
Consequently, the entries of $\hat{\bf V}$ corresponding to $\psi_1,\ldots,\psi_p$ describe the asymptotic variances of the estimated uniqueness variances themselves, rather than those of the computational parameters $\kappa_1,\ldots,\kappa_p$.
Their estimated standard errors are therefore
\begin{equation}
    \widehat{\operatorname{se}}(\hat{\psi}_j)
    =
    \left(
    \frac{\hat V_{\psi_j\psi_j}}{n}
    \right)^{1/2},
    \qquad j=1,\ldots,p.
\end{equation}

The parameters are subsequently transformed to the reporting scale on which the factor variances equal one.
Because the scaling matrix
\begin{equation}
    {\bf D}
    =
    \diag\left(
    {\bf\Phi}
    \right)^{1/2}
\end{equation}
depends on the estimated parameters, this transformation cannot be treated as a fixed linear transformation when calculating standard errors.

Define the mapping from the inferential parameter vector ${\bfs\theta}$ to the parameters reported on the standardized scale by
\begin{equation}
    {\bf h}({\bfs\theta})
    =
    \left[
    \begin{array}{c}
    \operatorname{vec}_{\mathrm{free}}({\bf\Lambda}_{s})\\
    \operatorname{vec}_{\mathrm{free}}({\bf B}_{s})\\
    \operatorname{vech}_{\mathrm{free}}({\bf\Phi}_{s})
    \end{array}
    \right],
    \label{eq:mdsem_reporting_map}
\end{equation}
where ${\bf\Lambda}_{s}$, ${\bf B}_{s}$, and ${\bf\Phi}_{s}$ are obtained from ${\bfs\theta}$ using the scaling transformation in (\ref{eq:mdsem_scaling_transformation}) with ${\bf D}=\diag({\bf\Phi})^{1/2}$.
Here, $\operatorname{vec}_{\mathrm{free}}$ extracts the free elements of a matrix, and $\operatorname{vech}_{\mathrm{free}}$ extracts the nonfixed distinct elements of a symmetric matrix.

Let
\begin{equation}
    \hat{\bf J}_{h}
    =
    \left.
    \frac{\partial{\bf h}({\bfs\theta})}
    {\partial{\bfs\theta}^{\p}}
    \right|_{{\bfs\theta}=\hat{\bfs\theta}_n}
    \label{eq:mdsem_reporting_jacobian}
\end{equation}
denote the Jacobian of this mapping.
The delta method gives the estimated asymptotic covariance matrix of the standardized parameter estimates as
\begin{equation}
    \hat{\bf V}_{h}
    =
    \hat{\bf J}_{h}
    \hat{\bf V}
    \hat{\bf J}_{h}^{\p}.
    \label{eq:mdsem_standardized_covariance}
\end{equation}
Consequently, the estimated standard errors of the elements of ${\bf h}(\hat{\bfs\theta}_n)$ are
\begin{equation}
    \widehat{\operatorname{se}}
    \left\{
    h_k(\hat{\bfs\theta}_n)
    \right\}
    =
    \left(
    \frac{\hat V_{h,kk}}{n}
    \right)^{1/2}.
    \label{eq:mdsem_standardized_standard_errors}
\end{equation}
In the numerical studies and empirical applications, the Jacobian $\hat{\bf J}_{h}$ is evaluated by numerical differentiation.
Thus, the dependence of the scaling transformation on the estimated parameters is incorporated into the standard errors of the reported factor loadings, path coefficients, and factor correlations.

\section{Numerical Simulations}
\label{sec:numerical_simulations}

Three simulation studies examined the principal statistical and computational consequences of the BW representation.
Simulation 1 evaluated whether the resulting MDSEM estimator exhibits finite-sample behavior comparable to conventional covariance-based SEM estimators under correct specification.
Simulation 2 assessed the confidence intervals enabled by the minimum discrepancy function representation.
Simulation 3 examined whether direct minimization of the BW discrepancy provides stable estimation under small samples and model misspecification.

\subsection{Simulation 1: Finite-Sample Behavior under Correct
Specification}
\label{subsec:simulation1}

\subsubsection{Design}

The first simulation examined the finite-sample behavior corresponding to the consistency result established above.
The proposed MDSEM estimator was compared with conventional covariance-based SEM estimators as the sample size increased.

The population model contained $p=16$ observed variables and $r=4$ factors.
Each factor was measured by four observed variables, and the loading matrix ${\bf\Lambda}^{*}$ had a complete simple structure without cross-loadings.
The nonzero loadings were specified as
\begin{eqnarray}
    (\lambda_{1,1}^{*},\ldots,\lambda_{4,1}^{*})
    &=& (0.6,0.7,0.6,0.7), \nonumber\\
    (\lambda_{5,2}^{*},\ldots,\lambda_{8,2}^{*})
    &=& (-0.6,-0.7,-0.6,-0.7), \nonumber\\
    (\lambda_{9,3}^{*},\ldots,\lambda_{12,3}^{*})
    &=& (0.3,0.4,0.3,0.4), \nonumber\\
    (\lambda_{13,4}^{*},\ldots,\lambda_{16,4}^{*})
    &=& (0.6,0.7,0.6,0.7),
\end{eqnarray}
and all other elements were fixed to zero.

The nonzero elements of the path coefficient matrix ${\bf B}^{*}$ were
\begin{equation}
    b_{4,1}^{*}=0.4,\qquad
    b_{4,2}^{*}=-0.4,\qquad
    b_{4,3}^{*}=0.4.
\end{equation}
Thus, the fourth factor had structural effects on the other three factors.
The structural disturbance correlation matrix ${\bf\Omega}^{*}$ had unit diagonal elements, and its $(1,2)$, $(1,3)$, and $(2,3)$ correlations, together with their symmetric counterparts, were set to $0.4$.
All other off-diagonal elements were zero.

The factor covariance matrix was
\begin{equation}
    {\bf\Phi}^{*}
    =
    ({\bf I}_r-{\bf B}^{*})^{-1\p}
    {\bf\Omega}^{*}
    ({\bf I}_r-{\bf B}^{*})^{-1}.
\end{equation}
The population parameters were transformed according to (\ref{eq:mdsem_scaling_transformation}) so that $\diag({\bf\Phi}^{*})={\bf I}_r$.
The population covariance matrix was then defined as
\begin{equation}
    {\bf\Sigma}^{*}
    =
    {\bf\Lambda}^{*}
    ({\bf I}_r-{\bf B}^{*})^{-1\p}
    {\bf\Omega}^{*}
    ({\bf I}_r-{\bf B}^{*})^{-1}
    {\bf\Lambda}^{*\p}
    +
    {\bf\Psi}^{*},
    \label{eq:simulation_covariance}
\end{equation}
where the diagonal uniqueness matrix ${\bf\Psi}^{*}$ was chosen so that $\diag({\bf\Sigma}^{*})={\bf I}_p$.

Ten sample-size conditions,
$
    n=100,200,\ldots,1000,
$
were considered.
For each condition, $S=100$ data sets were generated from the population covariance matrix ${\bf\Sigma}^{*}$.
Five estimation procedures were applied to each data set: the proposed MDSEM estimator, maximum likelihood (ML), generalized least squares (GLS), ML with the uniqueness variances constrained to be positive, and GLS with the same positivity constraints.
The four covariance-based SEM procedures were implemented using the
\texttt{sem} function in the \texttt{lavaan} package.

For each parameter block ${\bf A}\in\{{\bf\Lambda},{\bf B},{\bf\Psi},{\bf\Phi}\}$, let $\mathcal{J}_{\bf A}$ denote the index set of its evaluated elements.
Let $a_j^{*}$ be the population value and $\hat a_{s,j}$ the estimate from replication $s$, and define
\begin{equation}
    \bar{\hat a}_{j}
    =
    \frac{1}{S}\sum_{s=1}^{S}\hat a_{s,j}.
\end{equation}
The average sampling variance for block ${\bf A}$ was calculated as
\begin{equation}
    \operatorname{Var}(\hat{\bf A})
    =
    \frac{1}{|\mathcal{J}_{\bf A}|}
    \sum_{j\in\mathcal{J}_{\bf A}}
    \left\{
    \frac{1}{S}
    \sum_{s=1}^{S}
    (\hat a_{s,j}-\bar{\hat a}_{j})^2
    \right\},
    \label{eq:variance_of_estimates}
\end{equation}
and the average squared bias was
\begin{equation}
    \operatorname{Bias}^{2}(\hat{\bf A})
    =
    \frac{1}{|\mathcal{J}_{\bf A}|}
    \sum_{j\in\mathcal{J}_{\bf A}}
    (\bar{\hat a}_{j}-a_j^{*})^2.
    \label{eq:squared_bias}
\end{equation}
Their sum gives the blockwise mean squared error, apart from the minor finite-replication adjustment associated with using $S$ rather than $S-1$ in the variance calculation.
The primary interest was whether the squared bias and variance of the MDSEM estimates decreased with increasing sample size in a manner comparable to those of ML and GLS.

\subsubsection{Results}

Figure~\ref{fig:simulation1} presents the average squared bias and sampling variance for each parameter block across the sample-size conditions.
The proposed MDSEM estimator exhibited small squared bias and sampling variance comparable to those obtained by ML and GLS.
For all procedures, sampling variance was substantially larger than squared bias and decreased rapidly as the sample size increased.

Thus, under correct model specification, the BW formulation yielded an MDSEM estimator with finite-sample behavior comparable to conventional covariance-based SEM estimators.
The theoretical connection established above was therefore also reflected in the empirical behavior of the parameter estimates.

\begin{figure}[!htb]
    \centering
    \includegraphics[width=0.8\linewidth]{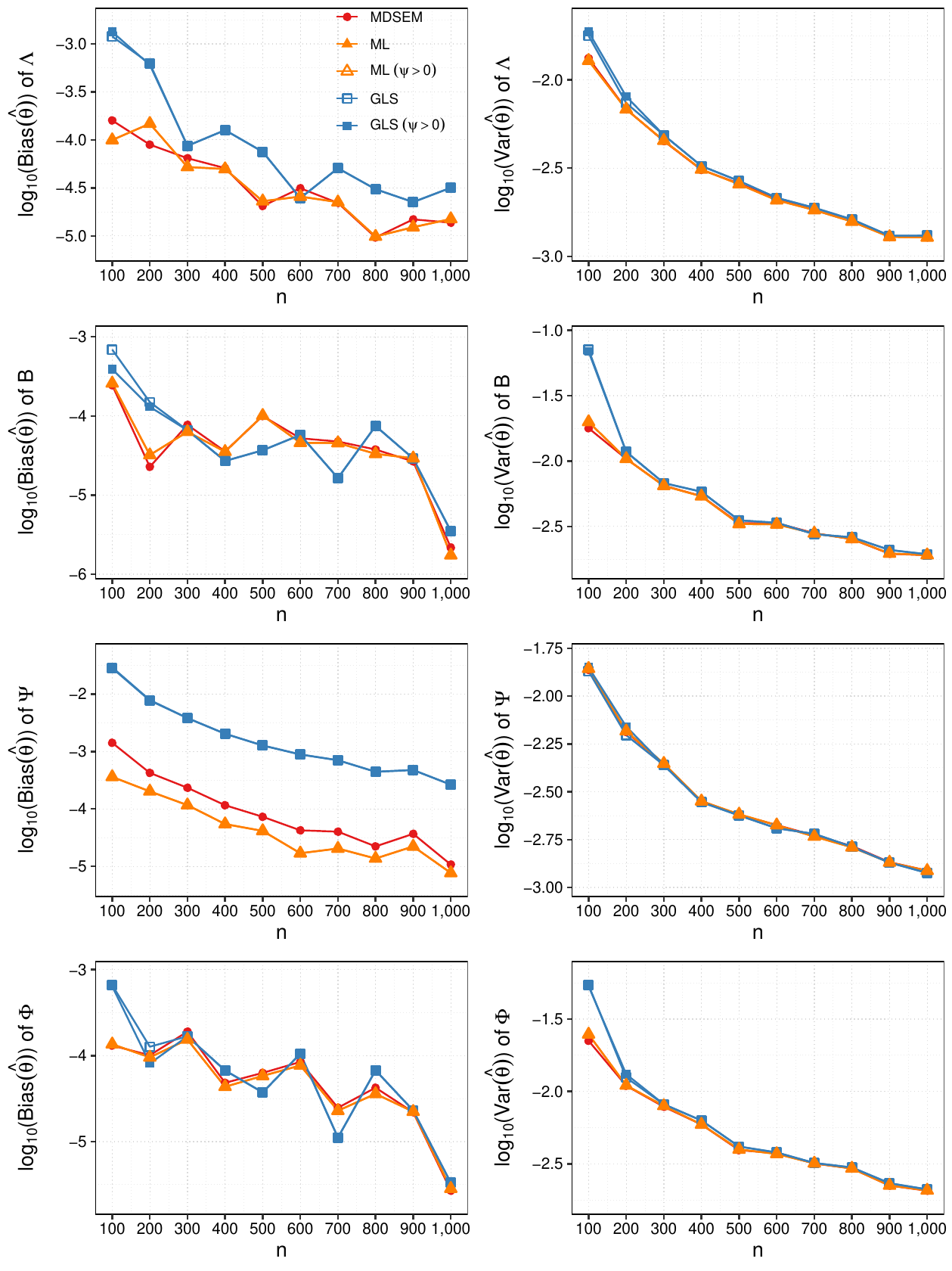}
    \caption{Average squared bias and sampling variance for each
    parameter block across sample-size conditions obtained by the five
    estimation procedures.}
    \label{fig:simulation1}
\end{figure}

\subsection{Simulation 2: Coverage of 95\% Confidence Intervals}
\label{subsec:simulation2}

\subsubsection{Design}

The second simulation evaluated the finite-sample accuracy of the standard errors derived from the asymptotic normality result.
Specifically, we examined whether the empirical coverage rates of nominal 95\% confidence intervals were close to 95\%.

The population model and parameter values were identical to those used in Simulation 1.
Four sample-size conditions,
$
    n=100,\ 300,\ 500,\ \text{and}\ 1000,
$
were considered.
For each condition, $S=1000$ data sets were generated under correct model specification.
The same five estimation procedures as in Simulation 1 were applied.

For each free parameter $\theta_j$, a nominal 95\% confidence interval was constructed as
\[
    \left[
    \hat{\theta}_j
    -
    1.96\,
    \widehat{\operatorname{se}}(\hat{\theta}_j),
    \quad
    \hat{\theta}_j
    +
    1.96\,
    \widehat{\operatorname{se}}(\hat{\theta}_j)
    \right].
\]
The empirical coverage rate was calculated as the proportion of replications in which this interval contained the corresponding population value $\theta_j^*$.
Coverage rates were then summarized separately for ${\bf\Lambda}$, ${\bf B}$, and ${\bf\Psi}$.

Coverage was evaluated only under correct model specification.
Under model misspecification, the data-generating parameters are not generally the asymptotic targets of the estimators.
Consequently, poor coverage could reflect misspecification bias rather than inaccurate standard errors, making it unsuitable as a direct assessment of the proposed inferential procedure.

\subsubsection{Results}

Figure~\ref{fig:simulation2} presents the empirical coverage rates of the nominal 95\% confidence intervals for each parameter block and sample-size condition.
For MDSEM, the coverage rates were generally close to the nominal level and did not exhibit systematic undercoverage.

These results indicate that the sandwich covariance estimator, together with the delta-method transformations, provides a reasonable finite-sample approximation to the sampling variability of the MDSEM estimates under correct model specification.
They also provide finite-sample support for the inferential framework obtained by representing MDSEM as a minimum discrepancy function estimator.

\begin{figure}[!htb]
    \centering
    \includegraphics[width=0.8\linewidth]{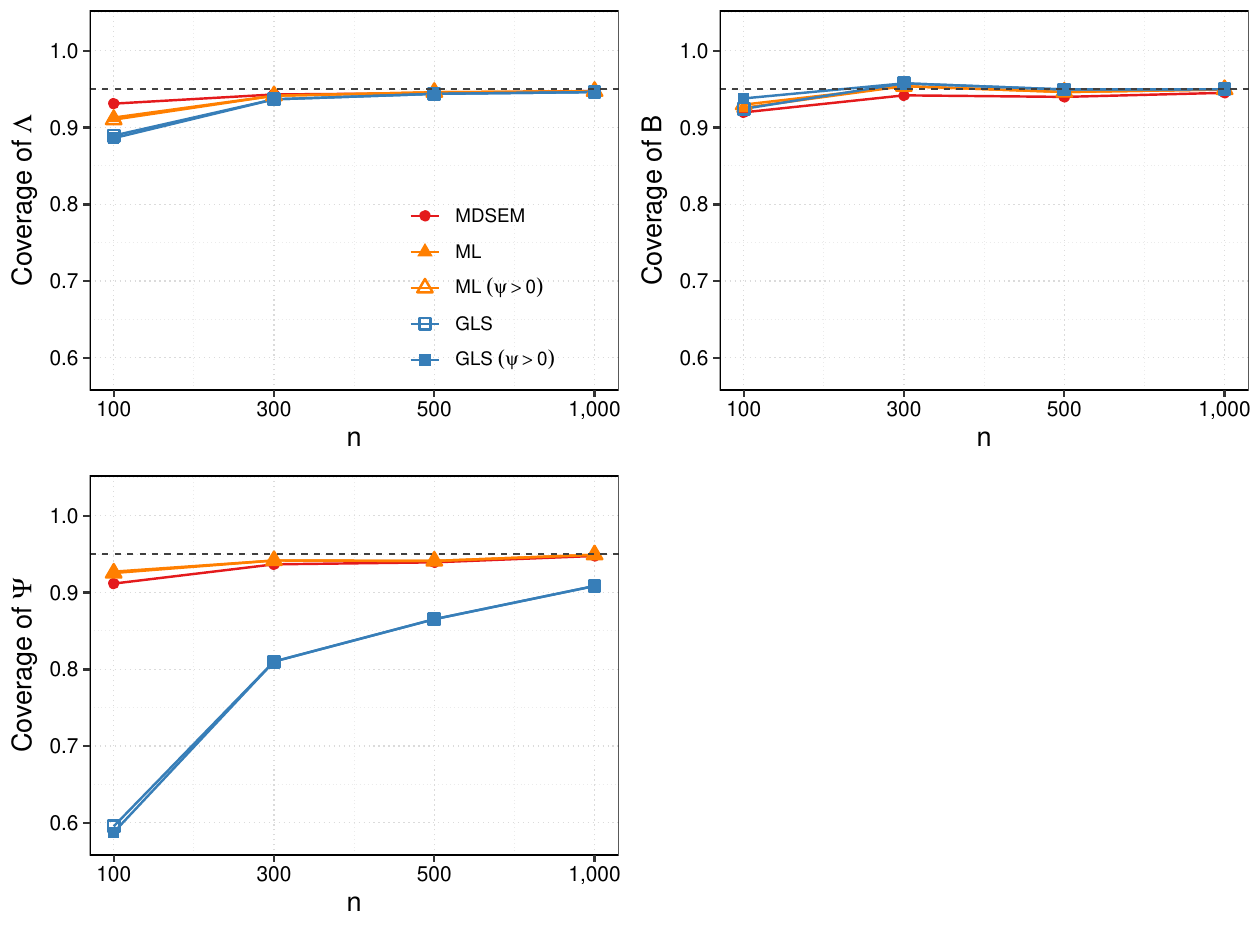}
    \caption{Empirical coverage rates of nominal 95\% confidence
    intervals for each parameter block across sample-size conditions
    obtained by the five estimation procedures.}
    \label{fig:simulation2}
\end{figure}

\subsection{Simulation 3: Performance under Small Samples and Model
Misspecification}
\label{subsec:simulation3}

\subsubsection{Design}

Model misspecification was introduced by adding cross-loadings, an omitted structural path, and residual covariances to the population model used in Simulation 1.
Specifically, cross-loadings were added for $(x_1,x_2)$ on Factor 2, $(x_5,x_6)$ on Factor 3, and $(x_9,x_{10})$ on Factor 4, with values $0.25$, $-0.25$, and $0.25$, respectively.
An additional structural path from Factor 2 to Factor 1 was set to $0.25$.
Residual covariances of $0.25$ were also added between three pairs of indicators, $(x_1,x_5)$, $(x_5,x_9)$, and $(x_9,x_{13})$.
These additional parameters were fixed to zero in the fitted model and therefore represented unmodeled measurement, structural, and residual relations.

The resulting population covariance matrix was constructed using (\ref{eq:simulation_covariance}) after incorporating these additional population parameters.
The population matrices were scaled so that $\diag({\bf\Sigma}^{*})={\bf I}_p$.

Eight sample-size conditions,
$
    n=30,\ 40,\ldots,\ 100,
$
were examined.
For each condition, $S=100$ data sets were generated, and the same five estimation procedures as in Simulations 1 and 2 were applied.

Numerical stability and estimation accuracy were evaluated using three criteria: the estimation failure rate, the 95th percentile of replication-wise root mean squared error (RMSE), and the condition number of the estimated Hessian matrix.
The 95th percentile of RMSE was used to assess upper-tail instability, that is, occasional large estimation errors that may not be apparent from mean RMSE.
The Hessian condition number was used as a diagnostic of local ill-conditioning and weak identification.
A large condition number indicates nearly flat or weakly identified parameter directions and may accompany convergence failures, improper solutions, or unstable standard-error estimates.

\subsubsection{Results}

Figure~\ref{fig:simulation3} presents the estimation failure rates and the 95th percentiles of RMSE for each parameter block.
Under small samples and model misspecification, ML and GLS frequently failed to converge or produced extreme estimation errors, whereas MDSEM converged in every replication and maintained comparatively stable upper-tail RMSE.
Figure~\ref{fig:simulation3_hessian} further shows that the Hessian condition numbers for MDSEM were substantially smaller than those for ML and GLS, especially at $n=30$.

These results indicate that direct optimization of the BW discrepancy can provide substantial numerical advantages under the challenging conditions considered here.
The main practical advantage of MDSEM in this setting is not uniformly superior average estimation accuracy, but resistance to severe numerical instability.

\begin{figure}[!htb]
    \centering
    \includegraphics[width=0.8\linewidth]{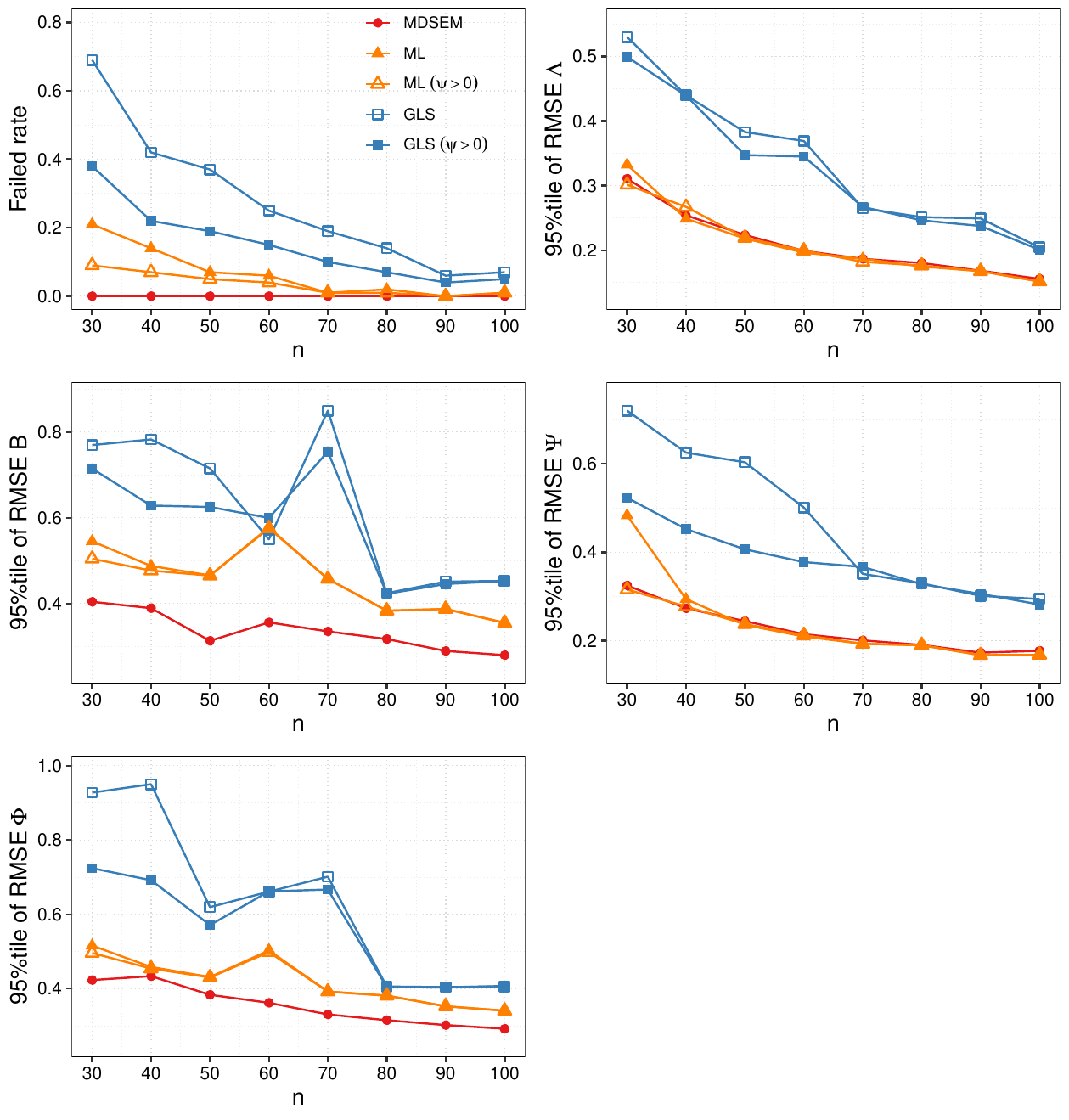}
    \caption{Estimation failure rates and 95th percentiles of RMSE for
    each parameter block across sample-size conditions obtained by the
    five estimation procedures.}
    \label{fig:simulation3}
\end{figure}

\begin{figure}[!htb]
    \centering
    \includegraphics[width=0.8\linewidth]
    {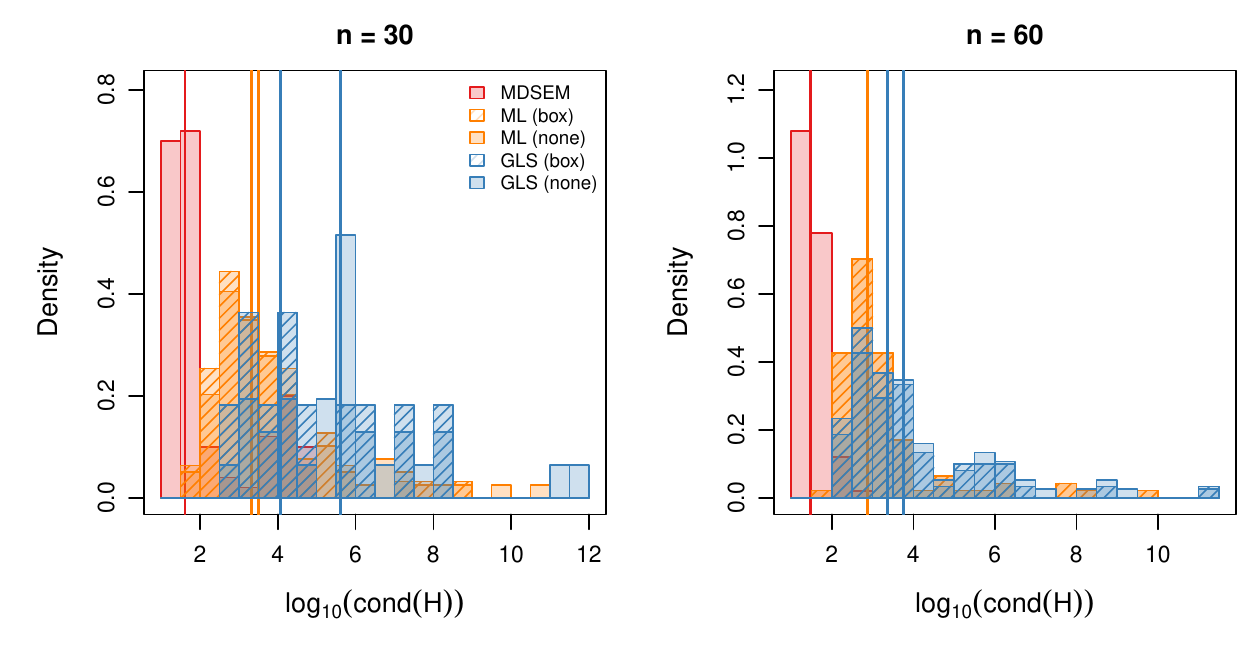}
    \caption{Distributions of Hessian condition numbers obtained by the
    five estimation procedures for $n=30$ and $n=60$. Vertical lines
    indicate the medians.}
    \label{fig:simulation3_hessian}
\end{figure}

\section{Empirical Comparison}
\label{sec:empirical_comparison}

As an empirical illustration, we analyzed the nine-variable intelligence-test data of Holzinger and Swineford (1939).
A three-factor SEM was fitted to the data using MDSEM, ML, and GLS.
The same model specification was used for all three estimation procedures.

Table~\ref{tab:example_param} presents the parameter estimates and standard errors.
Tables~\ref{tab:example_phi} and~\ref{tab:example_fit} present the estimated factor correlations and model-fit indices, respectively.
Overall, the estimates and standard errors obtained by MDSEM were close to those obtained by ML.
GLS produced somewhat different estimates for several parameters, but the general pattern of results was similar across the three procedures.
The model-fit indices obtained by MDSEM were also comparable to those obtained by ML.

Thus, in this empirical example, MDSEM produced results broadly consistent with those of conventional covariance-based SEM estimators.

\begin{table}[tb]
\centering
\caption{Factor loading and uniqueness estimates with standard errors, and path coefficient estimates with standard errors, obtained by MDSEM, ML, and GLS for the Holzinger--Swineford intelligence-test data.}
\label{tab:example_param}
\small
\begin{tabular}{llrrrrrrrrrrrr}
\toprule
& & \multicolumn{4}{c}{MDSEM}
& \multicolumn{4}{c}{ML}
& \multicolumn{4}{c}{GLS} \\
\cmidrule(lr){3-6}\cmidrule(lr){7-10}\cmidrule(lr){11-14}
& & \multicolumn{2}{c}{${\bf\Lambda}$}
& \multicolumn{2}{c}{${\bf\Psi}$}
& \multicolumn{2}{c}{${\bf\Lambda}$}
& \multicolumn{2}{c}{${\bf\Psi}$}
& \multicolumn{2}{c}{${\bf\Lambda}$}
& \multicolumn{2}{c}{${\bf\Psi}$} \\
\cmidrule(lr){3-4}\cmidrule(lr){5-6}
\cmidrule(lr){7-8}\cmidrule(lr){9-10}
\cmidrule(lr){11-12}\cmidrule(lr){13-14}
Variable & Factor
& est. & S.E. & est. & S.E.
& est. & S.E. & est. & S.E.
& est. & S.E. & est. & S.E. \\
\midrule
$x_1$ & $f_1$ & 0.79 & 0.08 & 0.37 & 0.10 & 0.77 & 0.07 & 0.40 & 0.08 & 0.67 & 0.07 & 0.40 & 0.07 \\
$x_2$ & $f_1$ & 0.42 & 0.06 & 0.81 & 0.08 & 0.42 & 0.07 & 0.82 & 0.07 & 0.32 & 0.07 & 0.76 & 0.07 \\
$x_3$ & $f_1$ & 0.58 & 0.06 & 0.65 & 0.07 & 0.58 & 0.07 & 0.66 & 0.07 & 0.49 & 0.07 & 0.59 & 0.06 \\
$x_4$ & $f_2$ & 0.85 & 0.05 & 0.27 & 0.04 & 0.85 & 0.05 & 0.27 & 0.04 & 0.83 & 0.05 & 0.25 & 0.03 \\
$x_5$ & $f_2$ & 0.85 & 0.04 & 0.27 & 0.04 & 0.85 & 0.05 & 0.27 & 0.04 & 0.83 & 0.05 & 0.23 & 0.03 \\
$x_6$ & $f_2$ & 0.84 & 0.05 & 0.29 & 0.04 & 0.84 & 0.05 & 0.30 & 0.04 & 0.83 & 0.05 & 0.30 & 0.04 \\
$x_7$ & $f_3$ & 0.53 & 0.06 & 0.69 & 0.07 & 0.57 & 0.06 & 0.67 & 0.07 & 0.59 & 0.06 & 0.50 & 0.06 \\
$x_8$ & $f_3$ & 0.68 & 0.06 & 0.52 & 0.07 & 0.72 & 0.07 & 0.48 & 0.07 & 0.70 & 0.06 & 0.47 & 0.06 \\
$x_9$ & $f_3$ & 0.73 & 0.06 & 0.45 & 0.07 & 0.66 & 0.06 & 0.56 & 0.07 & 0.71 & 0.06 & 0.45 & 0.06 \\
\midrule
\multicolumn{2}{l}{${\bf B}$} & est. & S.E. & & & est. & S.E. & & & est. & S.E. & &  \\
\midrule
$f_3 \sim f_1$ & & 0.44 & 0.09 & & & 0.43 & 0.09 & & & 0.60 & 0.12 & & \\
$f_3 \sim f_2$ & & 0.10 & 0.08 & & & 0.09 & 0.08 & & & 0.00 & 0.11 & & \\
\bottomrule
\end{tabular}
\end{table}

\begin{table}[tb]
\centering
\caption{Estimated factor correlations obtained by MDSEM, ML, and GLS for the Holzinger--Swineford intelligence-test data.}
\label{tab:example_phi}
\begin{tabular}{lrrr}
\toprule
Parameter & MDSEM & ML & GLS \\
\midrule
$f_1 \sim\sim f_2$ & 0.44 & 0.46 & 0.53 \\
$f_1 \sim\sim f_3$ & 0.48 & 0.47 & 0.60 \\
$f_2 \sim\sim f_3$ & 0.29 & 0.28 & 0.32 \\
\bottomrule
\end{tabular}
\end{table}

\begin{table}[tb]
\centering
\caption{Model-fit indices obtained by MDSEM, ML, and GLS for the Holzinger--Swineford intelligence-test data.}
\label{tab:example_fit}
\begin{tabular}{lrrr}
\toprule
Fit index & MDSEM & ML & GLS \\
\midrule
GFI   & 0.94 & 0.94 & 0.94 \\
AGFI  & 0.89 & 0.89 & 0.89 \\
SRMR  & 0.06 & 0.07 & 0.09 \\
RMSEA & 0.09 & 0.09 & 0.09 \\
\bottomrule
\end{tabular}
\end{table}

\section{Discussion}
\label{sec:discussion}
This study established an exact theoretical connection between the data-matrix formulation of MDSEM and the covariance-structure formulation of conventional SEM.
Specifically, by combining the measurement and structural equations and concentrating the resulting matrix-decomposition criterion over the auxiliary score matrix, we showed that this criterion is exactly equal to the squared BW distance between the sample covariance matrix and the model-implied covariance matrix.

This equivalence is the central theoretical contribution of the present study.
It shows that the MDSEM estimator can be interpreted as a minimum discrepancy function estimator within the conventional covariance-structure framework.
Consequently, standard minimum discrepancy theory yields consistency and asymptotic normality, as well as an asymptotic covariance matrix for calculating standard errors and confidence intervals.
The same representation also permits direct numerical minimization of the squared BW distance instead of alternating optimization over factor scores and model parameters.

The simulation studies supported the statistical and computational consequences of the BW representation.
Under correct model specification, the squared bias and sampling variance of the MDSEM estimates decreased as the sample size increased and were broadly comparable to those obtained by ML and GLS.
The empirical coverage rates of the nominal 95\% confidence intervals were also close to the target level, supporting the use of the sandwich covariance estimator and delta-method transformations for statistical inference.

More pronounced differences among the estimation procedures emerged under small samples and model misspecification.
ML and GLS frequently produced convergence failures or extreme estimation errors, whereas MDSEM converged in all examined replications and exhibited smaller upper-tail RMSE.
The Hessian condition numbers also tended to be substantially smaller for MDSEM.
These findings suggest that an important practical advantage of MDSEM is not uniformly superior average estimation accuracy, but resistance to severe numerical instability under unfavorable conditions.

A possible explanation for this stability is that the squared BW discrepancy does not directly involve either ${\bf\Sigma}^{-1}$ or ${\bf S}_{\bf X}^{-1}$.
Consequently, optimization based on this discrepancy may be less sensitive than ML or GLS to ill-conditioning in the model-implied covariance matrix or the sample covariance matrix.
This explanation is interpretive rather than a formal theoretical result, but it is consistent with the lower Hessian condition numbers, lower failure rates, and fewer extreme estimation errors observed for MDSEM in Simulation 3.

Several limitations should be noted.
First, the simulation studies considered a limited range of models, sample sizes, and forms of model misspecification; further research is needed to examine nonnormal data and larger, more complex models.
Second, the proposed asymptotic standard errors rely on regularity conditions, including local identification and a nonsingular Hessian, and may be less accurate under weak identification or near the boundary of the parameter space.

Further methodological development of MDSEM is required.
In particular, methods established in conventional SEM, including goodness-of-fit assessment, model comparison, the treatment of missing and ordinal data, and multiple-group analysis, have yet to be developed for MDSEM.
The BW formulation and asymptotic framework established in this study provide a basis for investigating these extensions.

\section*{Declaration}
\noindent
\textbf{Author Contributions}\ \ The author confirms sole responsibility for the following: study conception and mathematical development, data collection, analysis and interpretation of results, and manuscript preparation.\\
\noindent
\textbf{Funding}\ \ This research was supported by JSPS KAKENHI Grant Number 26K21186.\\
\noindent
\textbf{Conflict of Interest}\ \ The author has no conflicts of interest to disclose.\\
\noindent
\textbf{Data Availability}\ \ No empirical datasets were analyzed in this study. The simulated data were generated from the models described in the article.\\
\noindent
\textbf{Use of Generative Artificial Intelligence}\ \ The author used ChatGPT-5 to inspect spelling and grammatical errors during the writing process, but not to generate or revise any scientific content. The author independently verified the final text and remains fully responsible for the accuracy and integrity of the work.

\clearpage
\thispagestyle{empty}
\setcounter{equation}{0}
\renewcommand{\theequation}{S1.\arabic{equation}}


\begin{thebibliography}{99}

\bibitem[Adachi(2019)]{Adachi2019}
Adachi, K. (2019).
Factor analysis: Latent variable, matrix decomposition, and constrained uniqueness formulations.
{\it WIREs Computational Statistics}, 11(3), e1458.
doi: 10.1002/wics.1458.

\bibitem[Adachi and Trendafilov(2018)]{AdachiTrendafilov2018a}
Adachi, K. and Trendafilov, N. T. (2018).
Some mathematical properties of the matrix decomposition solution in factor analysis.
{\it Psychometrika}, 83(2), 407--424.
doi: 10.1007/s11336-017-9600-y.

\bibitem[Bartholomew et~al.(2011)]{bartholomew2011latent}
Bartholomew, D. J., Knott, M., and Moustaki, I. (2011).
{\it Latent Variable Models and Factor Analysis: A Unified Approach}.
John Wiley \& Sons.

\bibitem[Bentler(1980)]{bentler1980multivariate}
Bentler, P. M. (1980).
Multivariate analysis with latent variables: Causal modeling.
{\it Annual Review of Psychology}, 31(1), 419--456.

\bibitem[Bentler(1986)]{bentler1986structural}
Bentler, P. M. (1986).
Structural modeling and Psychometrika: An historical perspective on growth and achievements.
{\it Psychometrika}, 51(1), 35--51.

\bibitem[Bhatia et~al.(2019)]{BhatiaJainLim2019}
Bhatia, R., Jain, T., and Lim, Y. (2019).
On the Bures-Wasserstein distance between positive definite matrices.
{\it Expositiones Mathematicae}, 37(2), 165--191.
doi: 10.1016/j.exmath.2018.01.002.

\bibitem[Bollen(1989)]{Bollen1989}
Bollen, K. A. (1989).
{\it Structural Equations with Latent Variables}.
John Wiley \& Sons, New York.

\bibitem[Byrd et al.(1995)]{ByrdLuNocedalZhu1995}
Byrd, R. H., Lu, P., Nocedal, J., and Zhu, C. (1995).
A limited memory algorithm for bound constrained optimization.
\textit{SIAM Journal on Scientific Computing}, 16(5), 1190--1208.

\bibitem[Choi et~al.(2010)]{choi2010penalized}
Choi, J., Zou, H., and Oehlert, G. (2010).
A penalized maximum likelihood approach to sparse factor analysis.
{\it Statistics and Its Interface}, 3(4), 429--436.
doi: 10.4310/SII.2010.v3.n4.a1.

\bibitem[Cudeck and MacCallum(2007)]{cudeck2007factor}
Cudeck, R. and MacCallum, R. C. (Eds.) (2007).
{\it Factor Analysis at 100: Historical Developments and Future Directions}.
Routledge.

\bibitem[Freidlin(1968)]{Freidlin1968}
Freidlin, M. I. (1968).
On the factorization of non-negative definite matrices.
{\it Theory of Probability and Its Applications}, 13, 354--356.

\bibitem[Harman(1976)]{harman1976modern}
Harman, H. H. (1976).
{\it Modern Factor Analysis}.
University of Chicago Press.

\bibitem[Hirose and Yamamoto(2015)]{hirose2015sparse}
Hirose, K. and Yamamoto, M. (2015).
Sparse estimation via nonconcave penalized likelihood in factor analysis model.
{\it Statistics and Computing}, 25(5), 863--875.
doi: 10.1007/s11222-014-9458-0.

\bibitem[J\"oreskog(1969)]{joreskog1969}
J\"oreskog, K. G. (1969).
A general approach to confirmatory maximum likelihood factor analysis.
\textit{Psychometrika}, \textit{34}(2), 183--202.
doi: 10.1007/BF02289343

\bibitem[J\"oreskog(1970)]{joreskog1970}
J\"oreskog, K. G. (1970).
A general method for analysis of covariance structures.
\textit{Biometrika}, \textit{57}(2), 239--251.
doi: 10.1093/biomet/57.2.239

\bibitem[Kline(2023)]{kline2023}
Kline, R. B. (2023). 
{\it Principles and practice of structural equation modeling.}
Guilford publications.

\bibitem[Mulaik(2009)]{mulaik2009foundations}
Mulaik, S. A. (2009).
{\it Foundations of Factor Analysis}.
CRC Press.

\bibitem[Shapiro(1983)]{shapiro1983}
Shapiro, A. (1983).
Asymptotic distribution theory in the analysis of covariance structures:
A unified approach.
\textit{South African Statistical Journal}, \textit{17}, 33--81.

\bibitem[Shapiro(1984)]{shapiro1984}
Shapiro, A. (1984). 
A note on the consistency of estimators in the analysis of moment structures. 
\textit{British Journal of Mathematical and Statistical Psychology}, 37(1), 84–88.

\bibitem[Shapiro(1985a)]{shapiro1985a}
Shapiro, A. (1985a). 
Asymptotic distribution of test statistics in the analysis of moment structures under inequality constraints.
\textit{Biometrika}, 72(1), 133–144.

\bibitem[Shapiro(1985b)]{shapiro1985b}
Shapiro, A. (1985b). Asymptotic equivalence of minimum discrepancy function estimators to G.L.S. estimators.
\textit{South African Statistical Journal}, 19, 73–81.

\bibitem[So\v{c}an(2003)]{socan2003}
So\v{c}an, G. (2003).
The incremental value of minimum rank factor analysis.
Ph.D. thesis, University of Groningen.

\bibitem[Stegeman(2016)]{stegeman2016}
Stegeman, A. (2016).
A new method for simultaneous estimation of the factor model parameters, factor scores, and unique parts.
{\it Computational Statistics and Data Analysis}, 99, 189--203.

\bibitem[Terada(2025)]{terada2025}
Terada, Y. (2025).
Statistical properties of matrix decomposition factor analysis.
{\it arXiv preprint}, arXiv:2403.06968.

\bibitem[Yamashita(2024)]{Yamashita2024}
Yamashita, N. (2024).
Matrix decomposition approach for structural equation modeling as an alternative to covariance structure analysis and its theoretical properties.
{\it Structural Equation Modeling: A Multidisciplinary Journal}, 31(5), 817--834.
doi: 10.1080/10705511.2024.2342381.

\end{thebibliography}
\end{document}